\documentclass[journal,onecolumn,12pt]{IEEEtran}

 \usepackage{graphicx}
\usepackage{mathrsfs}
\usepackage{amsfonts}
\usepackage{graphicx}
\usepackage{amssymb}
\usepackage{amsmath}
\usepackage{cases}
\usepackage{caption}
\usepackage{hyperref}

\newtheorem{Thm}{Theorem}

\newtheorem{Lem}[Thm]{Lemma}
\newtheorem{Cor}[Thm]{Corollary}

\newtheorem{Exp}{Example}
\newcommand{\wt}{{\mathrm{wt}}}
\newcommand{\tr}{{\mathrm{Tr}}}
\newcommand{\C}{{\mathcal{C}}}

\newcommand{\F}{{\mathbb  F}}

\makeatletter
 \@addtoreset{equation}{section}
 
\makeatother

\begin{document}

\title{Complete weight enumerators of two classes of linear codes\thanks{The  authors are supported by by a National Key Basic Research Project of China (2011CB302400), National Natural Science Foundation of China (61379139),  the ``Strategic Priority Research Program" of the Chinese Academy of Sciences, Grant No. XDA06010701 and Foundation of NSSFC(No.13CTJ006).}}

\author{Qiuyan Wang, \thanks{Q. Wang is with the State Key Laboratory of Information Security, the Institute of Information Engineering, The Chinese Academy of Sciences, Beijing, China. Email: wangqiuyan@iie.ac.cn}
 Fei Li$^{*}$\thanks{F. Li is the corresponding author and  with School of Statistics and Applied Mathematics,
   Anhui University of Finance and Economics,
 Bengbu City, Anhui Province, China. Email: cczxlf@163.com},
 Kelan Ding\thanks{K. Ding is with the State Key Laboratory of Information Security, the Institute of Information Engineering, The Chinese Academy of Sciences, Beijing, China. Email: dingkelan@iie.ac.cn} and
 Dongdai Lin\thanks{ D. Lin is with the State Key Laboratory of Information Security, the Institute of Information Engineering, The Chinese Academy of Sciences, Beijing, China. Email: ddlin@iie.ac.cn}
 
}

\date{\today}
\maketitle

\begin{abstract}
Recently, linear codes with few weights have been constructed and extensively studied. In this paper, for an odd prime $p$, we determined the complete weight enumerator of two classes of $p$-ary linear codes constructed from defining set. Results show that the codes are at almost seven-weight linear codes and they may have applications in secret sharing  schemes.
\end{abstract}

\begin{keywords}
Linear code, complete weight enumerator, Gauss sum.
\end{keywords}

\section{Introduction}\label{sec-intro}

Throughout this paper, let $p$ ba an odd prime and let $q=p^{e}$ for some positive integer $e$. Let $\mathbb{F}_{p}$ denote the finite field with $p$ elements.
An $[n,k,d]$ linear code $\C$ is over $\mathbb{F}_{p}$ is a $k$-dimensional subspace of $\mathbb{F}_{p}^{n}$ with minimum distance $d$ \cite{HP03}.

Let $A_{i}$ be the number of codewords of weight $i$ in $\mathcal{C}$ of length $n$. The (Hamming) weight enumerator of $\mathcal{C}$ is defined by \cite{HP03}
$$
1+A_1x+A_2x^{2}+\cdots+A_{n}x^{n}.
$$
For $0\leq i\leq n$, the list $A_{i}$  is called the weight distribution or weight spectrum of $\mathcal{C}$. A code $\mathcal{C}$ is said to be a $t$-weight code if the number of nonzero $A_i$ with $1\leq i\leq n$ is equal to $t$. Clearly, the minimum distance of $\C$ can be derived from the weight distribution of the code $\C$. By error detection and error correction algorithms \cite{K07}, the weight distribution of a code can be applied to compute the error probability of error detection and correction. Thus, weight distribution is a significant  research topic in coding theory and was studied in \cite{BM72,CKNC12,CW84,CCD06,D09,D15,DY13,S12,V12}.

Let $ \mathbb{F}_p=\{0, 1, \ldots, p-1\}. $ For a codeword $\mathbf{c}=(c_{0},c_{1},\ldots,c_{n-1}) \in \C, $ the complete weight enumerator of $\mathbf{c}$ is the monomial $$w(\mathbf{c})=w_{0}^{t_{0}}w_{1}^{t_{1}}\cdots w_{p-1}^{t_{p-1}}$$ in the variables $w_{0}$, $w_{1}$, $\ldots$, $w_{p-1}.$ Here $t_{i}$ $(0\leq i \leq p-1) $ is the number of components of $\mathbf{c}$ which equal to $i.$ Then the complete weight enumerator of the linear code $\C$ is the homogeneous polynomial
$$\textrm{CWE}(\C)=\sum_{\mathbf{c}\in \C}w(\mathbf{c})$$
of degree $n$ (see \cite{MS77,MMS72}).

From definition, the complete weight enumerators of binary linear codes are just the weight enumerators. For nonbinary linear codes, the (Hamming) weight enumerators, which have been extensively investigated \cite{DGZ13,DD15,ZD14,DGZ13,ZLFH15},  can be obtained from the complete weight enumerators. Further more, the complete weight enumerators are closely related to the deception of some authentication codes constructed from linear codes \cite{DHKW07},  and used to compute the Walsh transform of monomial functions over finite fields \cite{HK06}. Thus, a great deal of research is devoted to the computation of the complete weight distribution of specific codes \cite{BLY15,BK91,K89,KN99,KN01,LYF15}.

Let $\mathbb{F}_{q}$ be the finite field with $q$ elements and $D=\{d_1,d_2,\ldots, d_n\}$ be  a nonempty  subset of $\mathbb{F}_{q}$. A generic construction of a linear code of length $n$ is given by
\begin{equation}\label{defcode1}
\C_{D}=\{\mathbf{c}_{x}=\left(\tr(xd_1), \tr(xd_2),\ldots, \tr(xd_{n})\right):x\in \mathbb{F}_{q}\},
\end{equation}
where $\tr$ denote the trace function from $\mathbb{F}_{q}$ onto $\mathbb{F}_{p}$ \cite{LN97}. The set $D$ is called the \textit{defining set} of $\C_{D}$. This construction technique is
 employed in lots of researches to get linear codes with few weights. The readers are referred to  \cite{DD14,QDX15,ZLFH15,DLN08,LYL14,QDX15} for more details.

 Naturally, a generalization of the code $\C_{D}$ of \eqref{defcode1} is defined by \cite{YY15}
\begin{equation}\label{defcode2}
 \overline{\C}_{D}=\{\left(\tr(xd_1)+u, \tr(xd_2)+u, \ldots, \tr(xd_{n})+u\right): u\in\F_{p},\  x\in \F_{q}\}.
\end{equation}

The objective of this paper is to present linear codes over $\mathbb{F}_{p}$ with at most seven weights using the above two construction methods. Further more, the complete weight enumerators of the two proposed linear codes are also  calculated. The codes  in this paper may have applications
in authentication codes \cite{DW05}, secret sharing schemes \cite{YD06} and
consumer electronics.

\section{The main results}
In this section, we only present the $p$-ary linear codes and introduce their parameters. The proofs of their parameters will be given later.

For $a\in \F_p,$ the defining  set is given by
\begin{eqnarray}\label{eq-1.1}
D_{a}=\{x\in \F_q^{*}: \tr(x^{p^{\alpha}+1})=a\},
\end{eqnarray}
where $\alpha$ is any natural number.
It should be remarked that, for $p=2$, the weight enumerator of $\C_{D_0}$ of \eqref{defcode1}  have been determined in \cite{DD14}. In this paper, for $a\in \mathbb{F}_p$ the complete weight enumerators of $\C_{D_a}$ of \eqref{defcode1} and $\overline{\C}_{D_a}$ of \eqref{defcode2} have been explicitly presented by using exponential sums.
\begin{Lem}[ \cite{C88},Lemma 2.6]\label{lem-2.1}
Let $d=\gcd(\alpha, e)$ and $p$ be odd. Then
$$
\gcd(p^{d}+1, p^{e}-1)=\left\{\begin{array}{ll}
                                2, & \textrm{if $e/d$  is odd},  \\
                                p^{d}+1, & \textrm{if $e/d$  is even}.
                              \end{array}
                              \right.
$$
\end{Lem}
Note that $\gcd(p^{\alpha}+1, p^{e}-1)=2$ leads to
$$
\{x^{p^{\alpha}+1}:x\in \mathbb{F}_{q}^{*}\}=\{x^{2}:x\in \mathbb{F}_{q}^{*}\}
$$
which means that
\begin{align*}
D_{0}&=\{x\in \mathbb{F}_{q}^{*}:\tr(x^{p^{\alpha}+1})=0\}\\
&=\{x\in \mathbb{F}_{q}^{*}:\tr(x^{2})=0\}.
\end{align*}
By Lemma \ref{lem-2.1}, if $e/d$ is odd, the code $\C_{D_0}$ of \eqref{defcode1} and the code $\C_{D}$ in \cite{DD15} are the same. Hence, we will assume  $e/d$ is even and $e=2m$ for a positive integer $m$.
The main results of this paper are given below.
\begin{table}[ht!]
\centering
\caption{The weight distribution of the codes of Theorem \ref{thm1} }\label{tal:weightdistribution1}
\begin{tabular}{|l|l|}
\hline
\textrm{Weight} $w$ \qquad& \textrm{Multiplicity} $A$   \\
\hline
0 \qquad&   1  \\
\hline
$(p-1)p^{e-2}$ \qquad&  $p^{e-1}-(p-1)p^{m-1}-1$  \\
\hline
$(p-1)(p^{e-2}-p^{m-1})$  \qquad& $(p-1)(p^{e-1}+p^{m-1})$  \\
\hline
\end{tabular}
\end{table}
\begin{Thm}\label{thm1}
Let $m\geq2$. If $m/d\equiv 1 \mod2,$ then the code $\C_{D_{0}}$ of \eqref{defcode1} is a $[p^{e-1}-(p-1)p^{m-1}-1, e]$ linear code with weight distribution in \autoref{tal:weightdistribution1}
and its complete weight enumerator is
\begin{align}
 &w_{0}^{p^{e-1}-(p-1)p^{m-1}-1}+(p^{e-1}-(p-1)p^{m-1}-1)
w_{0}^{p^{e-2}-(p-1)p^{m-1}-1}\prod_{i=1}^{p-1}w_{i}^{p^{e-2}} \nonumber \\
&+(p-1)(p^{e-1}+p^{m-1})
w_{0}^{p^{e-2}-1}\prod_{i=1}^{p-1}w_{i}^{p^{e-2}-p^{m-1}}. \nonumber \\
\nonumber
\end{align}

\begin{table}[ht]
\centering
\caption{The weight distribution of the codes of Corollary  \ref{thm1-1}}\label{tal:weightdistribution1-1}
\begin{tabular}{|l|l|}
\hline
\textrm{Weight} $b$ \qquad& \textrm{Multiplicity} $A$   \\
\hline
0 \qquad&   1  \\
\hline
$(p-1)p^{e-2}$ \qquad&  $p^{e-1}-(p-1)p^{m-1}-1$  \\
\hline
$(p-1)(p^{e-2}-p^{m-1})$  \qquad& $(p-1)(p^{e-1}+p^{m-1})$  \\
\hline
$p^{e-1}-(p-1)p^{m-1}-1$ \qquad&   $p-1$  \\
\hline
$(p-1)(p^{e-2}-p^{m-1})-1$ \qquad&  $(p-1)(p^{e-1}-(p-1)p^{m-1}-1)$  \\
\hline
$(p-1)p^{e-2}-(p-2)p^{m-1}-1$  \qquad& $(p-1)^{2}(p^{e-1}+p^{m-1})$  \\
\hline
\end{tabular}
\end{table}
\end{Thm}
\begin{Cor}\label{thm1-1}
Let the symbols and conditions be the same as Theorem \ref{thm1}.
The code $\overline{\C}_{D_{0}}$ of \eqref{defcode2} is a $[p^{e-1}-(p-1)p^{m-1}-1, e+1]$ linear code with weight distribution in \autoref{tal:weightdistribution1-1}
and its complete weight enumerator is
\begin{align}
 &\sum_{i=0}^{p-1}w_{i}^{p^{e-1}-(p-1)p^{m-1}-1} \nonumber \\
&+(p^{e-1}-(p-1)p^{m-1}-1)
\sum_{i=0}^{p-1}w_{i}^{p^{e-2}-(p-1)p^{m-1}-1}\prod_{j\neq i}w_{j}^{p^{e-2}} \nonumber \\
&+(p-1)(p^{e-1}+p^{m-1})
\sum_{i=0}^{p-1}w_{i}^{p^{e-2}-1}\prod_{j\neq i}w_{j}^{p^{e-2}-p^{m-1}}. \nonumber \\
\nonumber
\end{align}
\end{Cor}
\begin{table}[h!]
\centering
\caption{The weight distribution of the codes of Theorem \ref{thm2}.}\label{tal:weightdistribution2}
\begin{tabular}{|l|l|}
\hline
\textrm{Weight} $b$ \qquad& \textrm{Multiplicity} $A$   \\
\hline
0 \qquad&   1  \\
\hline
$(p-1)p^{e-2}+2p^{m-1}$ \qquad&  $\frac{1}{2}(p-1)(p^{e-1}+p^{m-1})$  \\
\hline
$(p-1)p^{e-2}$  \qquad& $p^{e}-\frac{1}{2}(p-1)(p^{e-1}+p^{m-1})-1$  \\
\hline
\end{tabular}
\end{table}
\begin{Thm}\label{thm2}
Let $g$ be a generator of $\mathbb{F}_{p}^{*}$ and $a\in \mathbb{F}_{p}^{*}. $ If $m/d\equiv 1 \mod2,$ then the code $\C_{D_{a}}$ \eqref{defcode1} is a $[p^{e-1}+p^{m-1}, e]$ linear code with weight distribution in table \autoref{tal:weightdistribution2}
and its complete weight enumerator is
\begin{align}
 &w_{0}^{p^{e-1}+p^{m-1}}+\left(p^{e-1}-(p-1)p^{m-1}-1\right)
w_{0}^{p^{e-2}+p^{m-1}}\prod_{i=1}^{p-1}w_{i}^{p^{e-2}} \nonumber \\
&+(p^{e-1}+p^{m-1})\sum_{\beta=1}^{\frac{p-1}{2}}
w_{0}^{p^{e-2}-\left(\frac{-1}{p}\right)p^{m-1}}w_{2g^{\beta}}^{p^{e-2}}w_{p-2g^{\beta}}^{p^{e-2}}
\prod_{i\neq0,\pm2g^{\beta}}w_{i}^{p^{e-2}-\left(\frac{i^{2}-4g^{2\beta}}{p}\right)p^{m-1}} \nonumber \\
&+(p^{e-1}+p^{m-1})\sum_{\beta=1}^{\frac{p-1}{2}}
w_{0}^{p^{e-2}+\left(\frac{-1}{p}\right)p^{m-1}}
\prod_{i=1}^{p-1}w_{i}^{p^{e-2}-\left(\frac{i^{2}-4g^{2\beta+1}}{p}\right)p^{m-1}}, \nonumber \\
\nonumber
\end{align}
where $\left(\frac{\cdot}{\cdot}\right)$ denotes the Legendre symbol.
\end{Thm}
\begin{table}[h!]
\centering
\caption{The weight distribution of the codes of Corollary \ref{thm2-2}}\label{tal:weightdistribution2-2}
\begin{tabular}{|l|l|}
\hline
\textrm{Weight} $b$ \qquad& \textrm{Multiplicity} $A$   \\
\hline
0 \qquad&   1  \\
\hline
$(p-1)p^{e-2}+2p^{m-1}$ \qquad&  $\frac{1}{2}(p-1)(p-2)(p^{e-1}+p^{m-1})$  \\
\hline
$(p-1)p^{e-2}$  \qquad& $p^{e}+\frac{1}{2}(p-1)(p-2)(p^{e-1}+p^{m-1})-1$  \\
\hline
$p^{e-1}+p^{m-1}$  \qquad& $p-1$  \\
\hline
$(p-1)p^{e-2}+p^{m-1}$ \qquad&   $(p-1)(2p^{e-1}-(p-2)p^{m-1}-1)$  \\
\hline
\end{tabular}
\end{table}
\begin{Cor}\label{thm2-2}
Let the symbols and conditions be the same as Theorem \ref{thm2}.
The code $\overline{\C}_{D_{a}}$ of \eqref{defcode2} is a $[p^{e-1}+p^{m-1}, e+1]$ linear code with weight distribution in \autoref{tal:weightdistribution2-2}
and its complete weight enumerator is
\begin{align}
 &\sum_{i=0}^{p-1}w_{i}^{p^{e-1}+p^{m-1}}+\left(p^{e-1}-(p-1)p^{m-1}-1\right)
\sum_{i=0}^{p-1}w_{i}^{p^{e-2}+p^{m-1}}\prod_{j\neq i}w_{j}^{p^{e-2}} \nonumber \\
&+(p^{e-1}+p^{m-1})\sum_{\beta=1}^{\frac{p-1}{2}}
\sum_{i=0}^{p-1}w_{i}^{p^{e-2}-\left(\frac{-1}{p}\right)p^{m-1}}w_{i+2g^{\beta}}^{p^{e-2}}w_{i-2g^{\beta}}^{p^{e-2}}
\prod_{j\neq i,i\pm2g^{\beta}}w_{j}^{p^{e-2}-\left(\frac{j^{2}-4g^{2\beta}}{p}\right)p^{m-1}} \nonumber \\
&+(p^{e-1}+p^{m-1})\sum_{\beta=1}^{\frac{p-1}{2}}
\sum_{i=0}^{p-1}w_{i}^{p^{e-2}+\left(\frac{-1}{p}\right)p^{m-1}}
\prod_{j\neq i}w_{j}^{p^{e-2}-\left(\frac{j^{2}-4g^{2\beta+1}}{p}\right)p^{m-1}}. \nonumber \\
\nonumber
\end{align}
\end{Cor}
\begin{Exp}
Let $(p,m,\alpha)=(3,3,1)$. If $a=0$, then the  code $\C_{D_0}$  has parameters $[224,6,144]$ with complete weight enumerator
$$
w_0^{224}+224w_0^{62}\prod_{i=1}^{2}w_{i}^{81}+504w_{0}^{80}\prod_{i\in 1}^{2}w_{i}^{72},
$$
and the $\overline{\C}_{D_0}$  has parameters $[224,7,143]$ with complete weight enumerator
$$
\sum_{i\in \mathbb{F}_{3}}w_{i}^{224}+224\sum_{i\in \mathbb{F}_{3}}w_{i}^{62}\prod_{j\neq i}w_{j}^{81}+504\sum_{i\in \mathbb{F}_{3}}w_{i}^{80}\prod_{j\neq i}w_{j}^{72}.
$$
If $a=1$, then the  code $\C_{D_1}$   has parameters $[252,6,162]$ with
complete weight enumerator
$$
w_{0}^{252}+476w_0^{90}\prod_{i=1}^{2}w_{i}^{81}+252w_0^{72}w_{1}^{90}w_{2}^{90},
$$
and the code $\overline{\C}_{D_1}$  has parameters $[252,7,162]$ with complete weight enumerator
$$
\sum_{i\in \mathbb{F}_{3}}w_{i}^{252}+476\sum_{i\in \mathbb{F}_{3}}w_{i}^{90}\prod_{j\neq i}w_{j}^{81}+252\sum_{i\in \mathbb{F}_{3}}w_{i}^{72}\prod_{j\neq i}w_{j}^{90}.
$$
\end{Exp}

\begin{table}[ht]
\centering
\caption{The weight distribution of the codes of Theorem \ref{thm3}.}\label{tal:weightdistribution3}
\begin{tabular}{|l|l|}
\hline
\textrm{Weight} $b$ \qquad& \textrm{Multiplicity} $A$   \\
\hline
0 \qquad&   1  \\
\hline
$(p-1)p^{e-2}-(p-1)^{2}p^{m+d-2}$ \qquad&  $p^{e}-p^{e-2d}$  \\
\hline
$(p-1)p^{e-2}$  \qquad& $p^{e-2d-1}-(p-1)p^{m-d-1}-1$ \\
\hline
$(p-1)p^{e-2}-(p-1)p^{m+d-1}$ \qquad&  $(p-1)(p^{e-2d-1}+p^{m-d-1})$  \\
\hline
\end{tabular}
\end{table}
\begin{Thm}\label{thm3}
Let $m>d+1$. If  $m/d\equiv 0 \mod2,$ then the code $\C_{D_{0}}$ of \eqref{defcode1} is a $[p^{e-1}-(p-1)p^{m+d-1}-1, e]$ linear code with weight distribution in \autoref{tal:weightdistribution3}
and its complete weight enumerator is
\begin{align}
&w_{0}^{p^{e-1}-(p-1)p^{m+d-1}-1}+(p-1)(p^{e-2d-1}+p^{m-d-1})
w_{0}^{p^{e-2}-1}\prod_{i=1}^{p-1}w_{i}^{p^{e-2}-p^{m+d-1}} \nonumber \\
&+(p^{e}-p^{e-2d})
w_{0}^{p^{e-2}-(p-1)p^{m+d-2}-1}\prod_{i=1}^{p-1}w_{i}^{p^{e-2}-(p-1)p^{m+d-2}} \nonumber \\
&+\left(p^{e-2d-1}-(p-1)p^{m-d-1}-1\right)
w_{0}^{p^{e-2}-(p-1)p^{m+d-1}-1}\prod_{i=1}^{p-1}w_{i}^{p^{e-2}}. \nonumber \\
\nonumber
\end{align}
\end{Thm}
\begin{table}[ht]
\centering
\caption{The weight distribution of the codes of Corollary \ref{thm3-3}.}\label{tal:weightdistribution3-3}
\begin{tabular}{|l|l|}
\hline
\textrm{Weight} $b$ \qquad& \textrm{Multiplicity} $A$   \\
\hline
0 \qquad&   1  \\
\hline
$(p-1)p^{e-2}-(p-1)^{2}p^{m+d-2}$ \qquad&  $p^{e}-p^{e-2d}$  \\
\hline
$(p-1)p^{e-2}$  \qquad& $p^{e-2d-1}-(p-1)p^{m-d-1}-1$ \\
\hline
$(p-1)p^{e-2}-(p-1)p^{m+d-1}$ \qquad&  $(p-1)(p^{e-2d-1}+p^{m-d-1})$  \\
\hline
$p^{e-1}-(p-1)p^{m+d-1}-1$ \qquad&   $p-1$  \\
\hline
$(p-1)p^{e-2}-(p-2)p^{m+d-1}-1$ \qquad&  $(p-1)^{2}(p^{e-2d-1}+p^{m-d-1})$  \\
\hline
$(p-1)(p^{e-2}-(p-1)p^{m+d-2})-1$  \qquad& $(p-1)(p^{e}-p^{e-2d})$ \\
\hline
$(p-1)(p^{e-2}-p^{m+d-1})-1$ \qquad&  $(p-1)(p^{e-2d-1}-(p-1)p^{m-d-1}-1)$  \\
\hline
\end{tabular}
\end{table}
\begin{Cor}\label{thm3-3}
Let the symbols and conditions  be the same as Theorem \ref{thm3}.
The code $\overline{\C}_{D_{0}}$ of \eqref{defcode2} is a $[p^{e-1}-(p-1)p^{m+d-1}-1, e+1]$ linear code with weight distribution in  \autoref{tal:weightdistribution3-3}
and its complete weight enumerator is
\begin{align}
 &\sum_{i=0}^{p-1}w_{i}^{p^{e-1}-(p-1)p^{m+d-1}-1} \nonumber \\
&+(p-1)(p^{e-2d-1}+p^{m-d-1})
\sum_{i=0}^{p-1}w_{i}^{p^{e-2}-1}\prod_{j\neq i}w_{j}^{p^{e-2}-p^{m+d-1}} \nonumber \\
&+(p^{e}-p^{e-2d})
\sum_{i=0}^{p-1}w_{i}^{p^{e-2}-(p-1)p^{m+d-2}-1}\prod_{j\neq i}w_{j}^{p^{e-2}-(p-1)p^{m+d-2}} \nonumber \\
&+\left(p^{e-2d-1}-(p-1)p^{m-d-1}-1\right)
\sum_{i=0}^{p-1}w_{i}^{p^{e-2}-(p-1)p^{m+d-1}-1}\prod_{j\neq i}w_{j}^{p^{e-2}}. \nonumber \\
\nonumber
\end{align}
\end{Cor}
\begin{table}[ht]
\centering
\caption{The weight distribution of the codes of Theorem \ref{Thm4}}\label{tal:weightdistribution4}
\begin{tabular}{|l|l|}
\hline
\textrm{Weight} $b$ \qquad& \textrm{Multiplicity} $A$   \\
\hline
0 \qquad&   1  \\
\hline
$(p-1)(p^{e-2}+p^{m+d-2})$ \qquad&  $p^{e}-p^{e-2d}$  \\
\hline
$(p-1)p^{e-2}$  \qquad& $\frac{1}{2}(p+1)p^{e-2d-1}-\frac{1}{2}(p-1)p^{m-d-1}-1$ \\
\hline
$(p-1)p^{e-2}+2p^{m+d-1}$ \qquad&  $\frac{1}{2}(p-1)(p^{e-2d-1}+p^{m-d-1})$  \\
\hline
\end{tabular}
\end{table}
\begin{Thm}\label{Thm4}
Let $ a\in \mathbb{F}_{p}^{*}$. If $m/d \equiv 0 \mod2,$ then the code $\C_{D_{a}}$ of \eqref{defcode1} is a $[p^{e-1}+p^{m+d-1}, e]$ linear code with weight distribution in \autoref{tal:weightdistribution4}
and its complete weight enumerator is
\begin{align*}
 &w_{0}^{p^{e-1}+p^{m+d-1}}+(p^{e}-p^{e-2d})
\prod_{i=0}^{p-1}w_{i}^{p^{e-2}+p^{m+d-2}}  \\
&+\left(p^{e-2d-1}-(p-1)p^{m-d-1}-1\right)w_{0}^{p^{e-2}+p^{m+d-1}}
\prod_{i=1}^{p-1}w_{i}^{p^{e-2}}  \\
&+(p^{e-2d-1}+p^{m-d-1})\sum_{\beta=1}^{\frac{p-1}{2}}
w_{0}^{p^{e-2}-\left(\frac{-1}{p}\right)p^{m+d-1}}w_{2g^{\beta}}^{p^{e-2}}w_{p-2g^{\beta}}^{p^{e-2}}
\prod_{i\neq0,\pm2g^{\beta}}w_{i}^{p^{e-2}-\left(\frac{i^{2}-4g^{2\beta}}{p}\right)p^{m+d-1}}  \\
&+(p^{e-2d-1}+p^{m-d-1})\sum_{\beta=1}^{\frac{p-1}{2}}
w_{0}^{p^{e-2}+\left(\frac{-1}{p}\right)p^{m+d-1}}
\prod_{i=1}^{p-1}w_{i}^{p^{e-2}-\left(\frac{i^{2}-4g^{2\beta+1}}{p}\right)p^{m+d-1}},
\end{align*}
where $\left(\frac{\cdot}{\cdot}\right)$ denote the Legendre symbol.
\end{Thm}
\begin{table}[ht]
\centering
\caption{The weight distribution of the codes of Corollary \ref{thm4-4}.}\label{tal:weightdistribution4-4}
\begin{tabular}{|l|l|}
\hline
\textrm{Weight} $b$ \qquad& \textrm{Multiplicity} $A$   \\
\hline
0 \qquad&   1  \\
\hline
$(p-1)(p^{e-2}+p^{m+d-2})$ \qquad&  $p(p^{e}-p^{e-2d})$  \\
\hline
$(p-1)p^{e-2}$  \qquad& $\frac{1}{2}(p^{2}-p+2)(p^{e-2d-1}+p^{m-d-1})-p^{m-d}-1$ \\
\hline
$(p-1)p^{e-2}+2p^{m+d-1}$ \qquad&  $\frac{1}{2}(p-1)(p-2)(p^{e-2d-1}+p^{m-d-1})$  \\
\hline
$p^{e-1}+p^{m+d-1}$ \qquad&  $p-1$  \\
\hline
$(p-1)p^{e-2}+p^{m+d-1}$  \qquad& $(p-1)(2p^{e-2d-1}-(p-2)p^{m-d-1}-1)$ \\
\hline
\end{tabular}
\end{table}
\begin{Cor}\label{thm4-4}
Let the symbols and conditions be the same as Theorem \ref{Thm4}.
The code $\overline{\C}_{D_{a}}$ of \eqref{defcode2} is a $[p^{e-1}+p^{m+d-1}, e+1]$ linear code with weight distribution in \autoref{tal:weightdistribution4-4}
and its complete weight enumerator is
\begin{align*}
 &\sum_{i=0}^{p-1}w_{i}^{p^{e-1}+p^{m+d-1}}+p(p^{e}-p^{e-2d})
\prod_{i=0}^{p-1}w_{i}^{p^{e-2}+p^{m+d-2}} \\
&+\left(p^{e-2d-1}-(p-1)p^{m-d-1}-1\right)\sum_{i=0}^{p-1}w_{i}^{p^{e-2}+p^{m+d-1}}
\prod_{j\neq i}w_{j}^{p^{e-2}} \\
&+(p^{e-2d-1}+p^{m-d-1})\sum_{\beta=1}^{\frac{p-1}{2}}
\sum_{i=0}^{p-1}w_{i}^{p^{e-2}-\left(\frac{-1}{p}\right)p^{m+d-1}}w_{i+2g^{\beta}}^{p^{e-2}}w_{i-2g^{\beta}}^{p^{e-2}}
\prod_{j\neq i,i\pm2g^{\beta}}w_{j}^{p^{e-2}-\left(\frac{j^{2}-4g^{2\beta}}{p}\right)p^{m+d-1}}  \\
&+(p^{e-2d-1}+p^{m-d-1})\sum_{\beta=1}^{\frac{p-1}{2}}
\sum_{i=0}^{p-1}w_{i}^{p^{e-2}+\left(\frac{-1}{p}\right)p^{m+d-1}}
\prod_{j\neq i}w_{j}^{p^{e-2}-\left(\frac{j^{2}-4g^{2\beta+1}}{p}\right)p^{m+d-1}}.
\end{align*}
\end{Cor}

\begin{Exp}
Let $(p,m,\alpha)=(3,4,1)$. If $a=0$, then the corresponding code $\C_{D_0}$ has parameters $[2024,8,1296]$ with complete weight enumerator
$$
w_0^{2024}+504w_0^{728}\prod_{i=1}^{2}w_{i}^{648}+5832w_0^{674}\prod_{i=1}^{2}w_{i}^{675}
+224w_0^{566}\prod_{i=1}^{2}w_{i}^{729},
$$
and the code $\overline{\C}_{D_0}$ has parameters $[2024,9,1295]$ with complete weight enumerator
$$
\sum_{i\in \mathbb{F}_{3}}w_{i}^{2024}+504\sum_{i\in \mathbb{F}_{3}}w_{i}^{728}\prod_{j\neq i}w_{j}^{648}+
5832\sum_{i\in \mathbb{F}_{3}}w_{i}^{674}\prod_{j\neq i}w_{j}^{675}
+224\sum_{i\in \mathbb{F}_{3}}w_{i}^{566}\prod_{j\neq i}w_{j}^{729}.
$$
If $a=1$,  then the  code $\C_{D_1}$ has parameters $[2268,8,1458]$ with complete weight enumerator
$$
w_{0}^{2268}+5832\prod_{i=0}^{2}w_{i}^{756}+
476w_{0}^{810}\prod_{i=1}^{2}w_{i}^{729}+252w_0^{648}w_1^{810}w_2^{810},
$$
and the code $\overline{\C}_{D_1}$ has parameters $[2268,9,1458]$ with  complete weight enumerator
$$
\sum_{i\in \mathbb{F}_{3}}w_{i}^{2268}+17496\prod_{i\in \mathbb{F}_{3}}w_{i}^{756}+496\sum_{i\in \mathbb{F}_{3}}w_{i}^{810}\prod_{j\neq i}w_{j}^{729}+252\sum_{i\in \mathbb{F}_{3}}w_{i}^{648}\prod_{j\neq i}w_{j}^{810}.
$$

\end{Exp}
\section{The proofs of the main results}
\label{}
Our task in this section is to prove results in Section $2$. Firstly,  we review some basic notations and results of group characters and
present some lemma which are needed for the proofs of the main results.
\subsection{Preliminaries}
We start with the additive character. A group homomorphism $\chi$ from $\mathbb{F}_{q}$ into
the complex numbers is called an additive character of $\mathbb{F}_{q}.$ Let $b\in \mathbb{F}_{q}$, the mapping
$$
\chi_{b}(c)=\zeta_{p}^{\tr(bc)} \ \textrm{for\ all }\ c\in \mathbb{F}_{q},
$$
defines an additive character of $\mathbb{F}_{q}, $ where $\zeta_{p}=e^{\frac{2\pi\sqrt{-1}}{p}}. $ The additive character $\chi_0$ is called \textit{trivial} and the other
characters $\chi_{b}$ with $b\in\mathbb{F}_{q}^{*}$ are called \textit{nontrivial}. The character $\chi_{1}$ is called the \textit{canonical additive character} of $\mathbb{F}_{q}$. And $\chi_{b}(x)=\chi_{1}(bx)$ for all $x\in \mathbb{F}_{q}$.

By the orthogonal property of additive characters, we have (\cite{LN97},  Theorem $5.4$),
$$
\sum_{x\in \mathbb{F}_{q}}\chi(x)=\left\{\begin{array}{ll}
                                           q, & \textrm{if\ } \chi \textrm{\ is\ trivial},  \\
                                           0, & \textrm{if\ } \chi \textrm{\ is\ nontrivial}.
                                         \end{array}
                                         \right.
$$

The \textit{multiplicative characters} of $\mathbb{F}_{q}$ are Characters over $\mathbb{F}_{q}^{*},$ which are given by
$$
\psi_{j}(g^{k})=\zeta_{p}^{2\pi\sqrt{-1}jk/(q-1)} \textrm{for }\ k=0,1,\ldots, q-2, \ 0\leq j\leq q-2.
$$
Here $g$ is a generator of $\mathbb{F}_{q}^{*}$ (\cite{LN97}).
The multiplicative character $\psi_{(q-1)/2}$ is called the \textit{quadratic character} of $\mathbb{F}_{q}$,
which is denoted by $\eta. $ And we assume that $ \eta(0)=0 $ in this paper.

We define the quadratic Gauss sum $G=G(\eta, \chi_1)$ over $\mathbb{F}_{q}$ by
$$G(\eta, \chi_1)=\sum_{x\in \mathbb{F}_{q}^{*}}\eta(x)\chi_1(x),$$
and the quadratic Gauss sum $\overline{G}=G(\overline{\eta}, \overline{\chi}_1)$ over $\mathbb{F}_{p}$ by
$$
G(\overline{\eta},\overline{\chi}_{1})=\sum_{x\in \mathbb{F}_{p}^{*}}\overline{\eta}(x)\overline{\chi}_1(x),
$$
where $\overline{\eta}$ and $\overline{\chi}_1$ denote the quadratic and canonical character of $\mathbb{F}_{p}.$

The explicit values of quadratic Gauss sums are given as follows.
\begin{Lem}[\cite{LN97}, Theorem $5.15$]\label{lem1}
Let the symbols be the same as before. Then
\begin{eqnarray*}
G(\eta, \chi_1)=(-1)^{(m-1)}\sqrt{-1}^{\frac{(p-1)^{2}m}{4}}\sqrt{q}, \ \
G(\overline{\eta}, \overline{\chi}_1)=\sqrt{-1}^{\frac{(p-1)^{2}}{4}}\sqrt{p}.
\end{eqnarray*}
\end{Lem}

\begin{Lem}[\cite{DD15}, Lemma7]\label{lem2}
Let the symbols be the same as before. Then
\begin{enumerate}
\item if $m\geq2$ is even, then $\eta(y)=1$ for each $y\in \mathbb{F}_{p}^{*}$;
\item if $m$ is odd, then $\eta(y)=\overline{\eta}(y)$ for each $y\in \mathbb{F}_{p}^{*}$.
\end{enumerate}
\end{Lem}

\begin{Lem}[\cite{LN97}, Theorem 5.33]\label{lem3}
Let $\chi$ be a nontrivial additive character of $\mathbb{F}_{q}$, and let $f(x)=a_2x^{2}+a_1x+a_0\in \mathbb{F}_{q}[x]$ with $a_2\neq0$. Then
$$
\sum_{x\in \mathbb{F}_{q}}\chi(f(x))=\chi\left(a_0-a_1^{2}(4a_2)^{-1}\right)\eta(a_2)G(\eta,\chi).
$$

For $a\in \mathbb{F}_{p}^{*}$, $b\in\mathbb{F}_{q},$ define
\begin{equation}\label{eq-a}
A(a)=\sum_{y\in \mathbb{F}_{p}^{*}}\zeta_{p}^{-ay}\sum_{x\in \mathbb{F}_{q}}\zeta_{p}^{y\tr(x^{p^{\alpha}+1})}
\end{equation}
and
\begin{equation}\label{eq-s}
S(a,b)=\sum_{x\in \mathbb{F}_{q}}\chi\left(ax^{p^{\alpha}+1}+bx\right).
\end{equation}
For $a$, $c\in\mathbb{F}_{p}$ and  $b\in\mathbb{F}_{q}^{*}$, define
\begin{equation}\label{eq-b0}
B(a,c)=\sum_{y\in \mathbb{F}_{p}^{*}}\sum_{z\in \mathbb{F}_{p}^{*}}\zeta_{p}^{-ay-cz}\sum_{x\in \mathbb{F}_{q}}\chi_1(yx^{p^{\alpha}+1}+bzx).
\end{equation}
\end{Lem}
\begin{Lem}[\cite{C88}, Theorem 2]\label{lem4}
 Let $e/d$ be even with $e=2m.$ Then
$$
S(a,0)=\left\{
   \begin{array}{ll}
   p^{m}, & \textrm{if} \ a^{\frac{q-1}{p^{d}+1}} \neq 1, \\
      -p^{m+d},&  \textrm{if} \ a^{\frac{q-1}{p^{d}+1}} = 1,\\
    -p^{m}, & \textrm{if} \ a^{\frac{q-1}{p^{d}+1}} \neq -1, \\
     p^{m+d}, & \textrm{if} \ a^{\frac{q-1}{p^{d}+1}} = -1,

\end{array}
\right.
  $$
\end{Lem}

\begin{Lem}[\cite{C98}, Theorem 1]\label{lem5}
 Let $q$ be odd and suppose $f(X)=a^{p^{\alpha}}X^{p^{2\alpha}}+aX$
is a permutation polynomial over $F_{q}.$ Let $x_{0}$ be the unique solution of the equation
$f(X)=-b^{p^{\alpha}}\  (b \neq 0).$ The evaluation of $S(a,b)$ partitions into the following two cases:
\begin{enumerate}
\item  If $e/d$ is odd,  then
$$
S(a,b)= \left\{
   \begin{array}{ll}
   (-1)^{e-1}\sqrt{q}\eta(-a)\overline{\chi}_{1}(ax_{0}^{p^{\alpha}+1}), & \textrm{if }  p \equiv 1 \mod4, \\
 (-1)^{e-1}i^{3e}\sqrt{q}\eta(-a)\overline{\chi}_{1}(ax_{0}^{p^{\alpha}+1}), & \textrm{if }  p \equiv 3 \mod4.
\end{array}
\right.
  $$

\item If $e/d$ is even, then $e=2m$, $a^{\frac{q-1}{p^{d}+1}} \neq (-1)^{\frac{m}{d}}$ and

$$ S(a,b)=(-1)^{\frac{m}{d}}p^{m}\overline{\chi}_{1}(ax_{0}^{p^{\alpha}+1}).$$
\end{enumerate}
\end{Lem}

\begin{Lem}[\cite{C98}, Theorem 2]\label{lem6}
 Let $q$ be odd and suppose $f(X)=a^{p^{\alpha}}X^{p^{2\alpha}}+aX$
is not a permutation polynomial over $F_{q}.$ Then for $b\neq 0$ we have $S(a,b)=0$ unless the
equation $f(X)=-b^{p^{\alpha}}$ is solvable. If the equation is solvable, with some solution $x_{0}$
say, then $S(a,b)=-(-1)^{\frac{m}{d}}p^{m+d}\overline{\chi}_{1}(ax_{0}^{p^{\alpha}+1}).$
\end{Lem}

\begin{Lem}\label{lem7}
Let the symbols be the same as before.
\begin{enumerate}
\item If $a=0$, then
$
A(0)=\left\{\begin{array}{ll}
                                                                          -(p-1)p^{m+d}, & \textrm{if\ }  m/d\equiv 0 \mod2, \\
                                                                          -(p-1)p^{m}, & \textrm{if\ } m/d\equiv 1 \mod2.
                                                                        \end{array}
                                                                        \right.
$
\item If $a\neq0$, then
$
A(a)=\left\{\begin{array}{ll}
                                                                          p^{m+d}, & \textrm{if\ } m/d\equiv 0 \mod2, \\
                                                                          p^{m}, & \textrm{if\ }   m/d\equiv 1 \mod2.
                                                                        \end{array}
                                                                        \right.
$
\end{enumerate}
\end{Lem}

\begin{proof}
Notice that $\sum_{x\in \mathbb{F}_{q}}\zeta_{p}^{y\tr(x^{p^{\alpha}+1})}=S(y,0)$.
 For $y\in \mathbb{F}_{p}^{*}$,  we know $y^{\frac{q-1}{p^{d}+1}} =1.$ Then the results can be obtained directly from Lemma \ref{lem4}.
\end{proof}

\begin{Lem}\label{lem8}
Let the symbols be the same as before.
If $m/d\equiv 1\mod2,$ then $X^{p^{2\alpha}}+ X=-b^{p^{\alpha}} $ has a solution $\gamma$ in $\mathbb{F}_{q}$ and
\begin{enumerate}
\item
$$
B(0,0)=\left\{
   \begin{array}{ll}
   -p^{m}(p-1)^{2}, & \textrm{if }\tr(\gamma^{p^{\alpha}+1})=0, \\
   p^{m}(p-1), \quad & \textrm{if }\tr(\gamma^{p^{\alpha}+1}) \neq 0;
\end{array}
  \right.
$$
\item for $c\neq0$, we have
$$
B(0,c)=\left\{\begin{array}{ll}
                                                                           (p-1)p^{m}, & \textrm{if } \tr(\gamma^{p^{\alpha}+1})=0, \\
                                                                           -p^{m}, & \textrm{if } \tr(\gamma^{p^{\alpha}+1}) \neq 0;
                                                                         \end{array}
                                                                         \right.
$$
\item for $a\neq0,$  we have
$$
B(a,0)=\left\{\begin{array}{ll}
                                                                          (p-1)p^{m}, & \textrm{if } \tr(\gamma^{p^{\alpha}+1})=0, \\
                                                                           -p^{m+1}\overline{\eta}(-a\tr(\gamma^{p^{\alpha}+1}))-p^{m}, & \textrm{if } \tr(\gamma^{p^{\alpha}+1}) \neq 0;
                                                                         \end{array}
                                                                         \right.
$$
\item for $ac\neq0,$ we have
$$
B(a,c)=\left\{\begin{array}{ll}
                                                                          -p^{m}, & \textrm{if } \tr(\gamma^{p^{\alpha}+1})=0, \\
                                                                          -p^{m}, & \textrm{if } \tr(\gamma^{p^{\alpha}+1})=c^{2}/(4a), \\
                                                                           -p^{m+1}\overline{\eta}(c^{2}-4a\tr(\gamma^{p^{\alpha}+1}))-p^{m}, & \textrm{otherwise}.
                                                                         \end{array}
                                                                         \right.
$$
\end{enumerate}
\end{Lem}
\begin{proof}
If $m/d\equiv 1 \mod2,$ by Theorem 4.1 in \cite{CW84}, we get the equation $X^{p^{2\alpha}}+X=0$ has no solution in $\mathbb{F}_{q}^{*}.$
 Then $X^{p^{2\alpha}}+ X$ is a permutation polynomial over $\mathbb{F}_{q}$ and $X^{p^{2\alpha}}+ X=-b^{p^{\alpha}} \textrm{has a unique solution }\gamma\ \text{in }\mathbb{F}_{q}.$ Thus $y^{-1}z\gamma$ is the unique solution in $\mathbb{F}_{q}$ of $y^{p^{\alpha}}X^{p^{2\alpha}}+ yX=-(bz)^{p^{\alpha}}.$ By Lemma \ref{lem5}, $S(y,bz)=-p^{m}\overline{\chi}_{1}\left(y(y^{-1}z\gamma)^{p^{\alpha}+1}\right).$ Therefore,
\begin{align*}
 B(a,c)&=-p^{m}\sum_{y\in \mathbb{F}_{p}^{*}}\sum_{z\in \mathbb{F}_{p}^{*}}\zeta_{p}^{-ay-cz}\overline{\chi}_{1}\left(y(y^{-1}z\gamma)^{p^{\alpha}+1}\right)\\
&=-p^{m}\sum_{y\in \mathbb{F}_{p}^{*}}\sum_{z\in \mathbb{F}_{p}^{*}}\zeta_{p}^{-ay-cz}\zeta_{p}^{-\frac{z^{2}\tr(\gamma^{p^{\alpha}+1})}{y}}.
\end{align*}
 If $\tr(\gamma^{p^{\alpha}+1})=0,$ the corresponding results follow from the orthogonal property of additive characters easily. \\
If  $\tr(\gamma^{p^{\alpha}+1})\neq0,$ we have
$$B(0,0)=-p^{m}\sum_{z\in \mathbb{F}_{p}^{*}}\sum_{y\in \mathbb{F}_{p}^{*}}\zeta_{p}^{-\frac{z^{2}\tr(\gamma^{p^{\alpha}+1})}{y}}=-p^{m}\sum_{z\in \mathbb{F}_{p}^{*}}\sum_{y\in \mathbb{F}_{p}^{*}}\zeta_{p}^{y}=(p-1)p^{m}.$$
$$B(0,c)=-p^{m}\sum_{z\in \mathbb{F}_{p}^{*}}\sum_{y\in \mathbb{F}_{p}^{*}}\zeta_{p}^{-cz}\zeta_{p}^{-\frac{z^{2}\tr(\gamma^{p^{\alpha}+1})}{y}}=-p^{m}\sum_{z\in \mathbb{F}_{p}^{*}}\zeta_{p}^{-cz}\sum_{y\in \mathbb{F}_{p}^{*}}\zeta_{p}^{y}=-p^{m}.$$
By Lemma \ref{lem3}, we get
\begin{align*}
 B(a,0)&=-p^{m}\sum_{y\in \mathbb{F}_{p}^{*}}\zeta_{p}^{-ay}\sum_{z\in \mathbb{F}_{p}^{*}}\zeta_{p}^{-\frac{z^{2}\tr(\gamma^{p^{\alpha}+1})}{y}}\\
&=-p^{m}\sum_{y\in \mathbb{F}_{p}^{*}}\zeta_{p}^{-ay}\sum_{z\in \mathbb{F}_{p}^{*}}\overline{\chi}_{1}\left(-\frac{\tr(\gamma^{p^{\alpha}+1})}{y}z^{2}\right)\\
&=-p^{m}\sum_{y\in \mathbb{F}_{p}^{*}}\zeta_{p}^{-ay}\sum_{z\in \mathbb{F}_{p}}\overline{\chi}_{1}\left(-\frac{\tr(\gamma^{p^{\alpha}+1})}{y}z^{2}\right)+p^{m}\sum_{y\in \mathbb{F}_{p}^{*}}\zeta_{p}^{-ay}\\
&=-p^{m}\sum_{y\in \mathbb{F}_{p}^{*}}\zeta_{p}^{-ay}\overline{\chi}_{1}(0)\overline{\eta}\left(-\frac{\tr(\gamma^{p^{\alpha}+1})}{y}\right)\overline{G}-p^{m}\\
&=-p^{m}\overline{\eta}(a\tr(\gamma^{p^{\alpha}+1}))\overline{G}\sum_{y\in \mathbb{F}_{p}^{*}}\zeta_{p}^{-ay}\overline{\eta}(-ay)-p^{m}\\
&=-p^{m}\overline{\eta}(a\tr(\gamma^{p^{\alpha}+1}))\overline{G}^{2}-p^{m}\\
&=-p^{m+1}\overline{\eta}(-a\tr(\gamma^{p^{\alpha}+1}))-p^{m}.
\end{align*}
Also by Lemma \ref{lem3}, we have
\begin{align}\label{eq-b}
B(a,c)&=-p^{m}\sum_{y\in \mathbb{F}_{p}^{*}}\zeta_{p}^{-ay-cz}\sum_{z\in \mathbb{F}_{p}^{*}}\zeta_{p}^{-\frac{z^{2}\tr(\gamma^{p^{\alpha}+1})}{y}} \nonumber\\
&=-p^{m}\sum_{y\in \mathbb{F}_{p}^{*}}\zeta_{p}^{-ay}\sum_{z\in \mathbb{F}_{p}}\overline{\chi}_{1}\left(-\frac{\tr(\gamma^{p^{\alpha}+1})}{y}z^{2}-cz\right)+p^{m}\sum_{y\in \mathbb{F}_{p}^{*}}\zeta_{p}^{-ay}\nonumber\\
&=-p^{m}\sum_{y\in \mathbb{F}_{p}^{*}}\zeta_{p}^{-ay}\overline{\chi}_{1}\left(\frac{yc^{2}}{4\tr(\gamma^{p^{\alpha}+1})}\right)
\overline{\eta}\left(-\frac{\tr(\gamma^{p^{\alpha}+1})}{y}\right)\overline{G}-p^{m}\nonumber \\
&=-p^{m}\overline{\eta}\left(-\tr(\gamma^{p^{\alpha}+1})\right)\overline{G}\sum_{y\in \mathbb{F}_{p}^{*}}\zeta_{p}^{\frac{c^{2}-4a\tr(\gamma^{p^{\alpha}+1})}{4\tr(\gamma^{p^{\alpha}+1})}y}\overline{\eta}(y)-p^{m}.
\end{align}
If $\tr(\gamma^{p^{\alpha}+1})=c^{2}/(4a)$, by \eqref{eq-b}, we have $B(a,c)=-p^{m}.$  Otherwise, we obtain
\begin{align*}
B(a,c)&=-p^{m}\overline{\eta}\left(4a\tr(\gamma^{p^{\alpha}+1})-c^{2}\right)\overline{G}\sum_{y\in \mathbb{F}_{p}^{*}}\zeta_{p}^{\frac{c^{2}-4a\tr\left(\gamma^{p^{\alpha}+1}\right)}
{4\tr(\gamma^{p^{\alpha}+1})}y}\overline{\eta}\left(\frac{c^{2}-4a\tr(\gamma^{p^{\alpha}+1})}
{4\tr(\gamma^{p^{\alpha}+1})}y\right)-p^{m}\\
&=-p^{m}\overline{\eta}\left(-\left(c^{2}-4a\tr(\gamma^{p^{\alpha}+1})\right)\right)\overline{G}^{2}-p^{m}\\
&=-p^{m+1}\overline{\eta}\left(c^{2}-4a\tr(\gamma^{p^{\alpha}+1})\right)-p^{m}.
\end{align*}
 The proof of this lemma is completed.
\end{proof}

\begin{Lem}\label{lem9}
Let the symbols be the same as  Lemma \ref{lem8}. Let $m/d \equiv 0\mod2.$ If $X^{p^{2\alpha}}+ X=-b^{p^{\alpha}} \text{has no solution in } \mathbb{F}_{q}, $ then
$$B(a,c)=0.$$
Suppose $X^{p^{2\alpha}}+ X=-b^{p^{\alpha}}$ has a solution $\gamma$ in $\mathbb{F}_{q}$, we have
\begin{enumerate}
\item
$$
B(0,0)=\left\{
   \begin{array}{ll}
   -p^{m+d}(p-1)^{2}, &  \textrm{if } \tr(\gamma^{p^{\alpha}+1})=0, \\
   p^{m+d}(p-1), & \textrm{if } \tr(\gamma^{p^{\alpha}+1})\neq0.
\end{array}
  \right.
$$

\item if $c\neq0$, then
$$
B(0,c)=\left\{\begin{array}{ll}
   (p-1)p^{m+d}, & \textrm{if }  \tr(\gamma^{p^{\alpha}+1})=0, \\
   -p^{m+d}, & \textrm{if }  \tr(\gamma^{p^{\alpha}+1})\neq0.
   \end{array}
   \right.
$$
\item if $a\neq0$, then
$$
B(a,0)=\left\{\begin{array}{ll}
   (p-1)p^{m+d}, & \textrm{if }  \tr(\gamma^{p^{\alpha}+1})=0, \\
   -p^{m+d+1}\overline{\eta}\left(-a\tr(\gamma^{p^{\alpha}+1})\right)-p^{m+d}, & \textrm{if }  \tr(\gamma^{p^{\alpha}+1})\neq0, \\
   \end{array}
   \right.
$$
\item if $ac\neq0,$ then
$$
B(a,c)=\left\{\begin{array}{ll}
   -p^{m+d}, & \textrm{if }  \tr(\gamma^{p^{\alpha}+1})=0, \\
   -p^{m+d}, & \textrm{if } \tr(\gamma^{p^{\alpha}+1})=c^{2}/(4a),\\
   -p^{m+d+1}\overline{\eta}\left(c^{2}-4a\tr(\gamma^{p^{\alpha}+1})\right)-p^{m+d}, & \textrm{otherwise}.
   \end{array}
   \right.
$$
\end{enumerate}
\end{Lem}

\begin{proof}
 If $m/d \equiv 0\mod2,$ by Theorem 4.1 in \cite{CW84}, we know that $X^{p^{2\alpha}}+X$ is not a permutation polynomial.
Then using Lemma \ref{lem6}, the proof is similar to that of Lemma \ref{lem8}. We omit the details.
\end{proof}

\begin{Lem}\label{lem10}
Set $f(X)= X^{p^{2\alpha}}+X$ and
$$ S=\{b\in \mathbb{F}_{q}: f(x)=-b^{p^{\alpha}} \text{is solvable in}\ \mathbb{F}_{q} \}.$$
If $m/d$ is even, then $|S|=p^{e-2d}$.
\end{Lem}
\begin{proof}
Note that both $e/d$ and $m/d$ are even. By Theorem 4.1 in \cite{CW84}, it is easy to see that $f(X)=0 $
has $p^{2d}$ solutions in $\mathbb{F}_{q}.$ So does the equation $f(X)=-b^{p^{\alpha}}$ with $ b\in S.$ For
$ b_{1},b_{2} \in \mathbb{F}_{q}$ and $b_{1}\neq b_{2}, $ it is not hard to know that the two equations $ f(X)=-b_{1}^{p^{\alpha}} $
and $f(X)=-b_{2}^{p^{\alpha}} $ have no common solutions. Furthermore, for each $a \in \mathbb{F}_{q},$ there must exist
$b\in \mathbb{F}_{q} $ such that $f(a)=-b^{p^{\alpha}}$, since $X^{p^{\alpha}}$ is a permutation polynomial over $\mathbb{F}_{q}.$
So the desired result can be concluded with the above discussions.
\end{proof}
\subsection{Proofs of theorems in Section A}

For $a\in \mathbb{F}_{p},$ recall that we set
$$
 D_{a}=\{x\in \mathbb{F}_{q}^{*}: \tr(x^{p^{\alpha}+1})=a\}.
$$
Let
$$
n_{a}=\left\{\begin{array}{ll}
                                                                          |D_{a}\cup \{0\}|, & \textrm{if\ }  a=0,  \\
                                                                          |D_{a}|, & \textrm{if\ }  a\neq 0.
                                                                        \end{array}
                                                                        \right.
$$
From the definition of $n_a$, we know
\begin{align*}
n_{a}& =\frac{1}{p}\sum_{x \in F_{q}}\left(\sum_{y \in F_{p}}\zeta_{p}^{y\tr(x^{p^{\alpha}+1})-ay}\right)\\
&= p^{e-1}+\frac{1}{p}\sum_{y \in F_{p}^{\ast}}\zeta_{p}^{-ay}\sum_{x \in F_{q}}\zeta_{p}^{y\tr(x^{p^{\alpha}+1})}.
\end{align*}
Therefore,
\begin{equation} \label{eq-3.1}
n_{a}=p^{e-1}+p^{-1}A(a),
\end{equation}
where $A(a)$ is defined by \eqref{eq-a}.

For $a$, $c\in \mathbb{F}_{p}$ and $b\in \mathbb{F}_{q}^{*}$, define
$$N_{b}(a,c)=\{x\in \mathbb{F}_{q}: \tr(x^{p^{\alpha}+1})=a\textrm{\ and\ } \tr(bx)=c\}.$$
Let $\wt(c_b)$ denote the Hamming weight of the  codeword
 $
 \mathbf{c}_b
 $  $(b\in \mathbb{F}_{q}^{*})$
 of the code $\C_{D_a}.$ It is not difficult to see that
\begin{equation} \label{eq-3.2}
\wt(c_b)=n_a-|N_{b}(a,c)|.
\end{equation}.

From definition, for $b\in \mathbb{F}_{q}^{*},$ and $a$, $c\in \mathbb{F}_{p}$ we have
\begin{align}\label{eq-3.3}
|N_{b}(a,c)| &= p^{-2}\sum_{x\in \mathbb{F}_{q}}\left(\sum_{y \in F_{p}}\zeta_{p}^{y\tr(x^{p^{\alpha}+1})-ay}\right)
\left(\sum_{z \in F_{p}}\zeta_{p}^{z\tr(bx)-cz}\right)  \nonumber \\
 &=p^{-2}\sum_{x\in \mathbb{F}_{q}}\left(1+\sum_{y \in F_{p}^{\ast}}\zeta_{p}^{y\tr(x^{p^{\alpha}+1})-ay}\right)
\left(1+\sum_{z \in F_{p}^{\ast}}\zeta_{p}^{z\tr(bx)-cz}\right)\nonumber \\
&= p^{e-2} + p^{-2}\sum_{y \in F_{p}^{\ast}}\sum_{x\in \mathbb{F}_{q}}\zeta_{p}^{y\tr(x^{p^{\alpha}+1})-ay}
+ p^{-2}\sum_{z \in F_{p}^{\ast}}\sum_{x\in \mathbb{F}_{q}}\zeta_{p}^{z\tr(bx)-cz} \nonumber \\
&\quad +
p^{-2}\sum_{y \in F_{p}^{\ast}}\sum_{z \in F_{p}^{\ast}}\sum_{x\in \mathbb{F}_{q}}\zeta_{p}^{\tr(yx^{p^{\alpha}+1}+bzx)-ay-cz} \nonumber\\
&= p^{e-2} + p^{-2}A(a) +
p^{-2}B(a,c)
\end{align}
where $A(a)$ and $B(a,c)$ are defined by \eqref{eq-a} and \eqref{eq-b0}, respectively.

Our task in the sequel  is to calculate $n_{a}$, $|N_{b}(a,c)|$ and give the proofs of the main results.

\subsubsection{The first case}

 $a=0$  and $m/d\equiv 1 \mod2.$

 Let $\gamma$ be the unique solution of
equation $X^{p^{2\alpha}}+X=-b^{p^{\alpha}}.$ By \eqref{eq-3.3}, Lemmas \ref{lem7} and \ref{lem8}, we have the following two lemmas.

\begin{Lem}\label{lem11}
Let $b\in \mathbb{F}_{q}^{*},$ then
$$
|N_{b}(0,0)|=\left\{\begin{array}{ll}
                                                                          p^{e-2}-(p-1)p^{m-1}, & \textrm{if\ } \tr(\gamma^{p^{\alpha}+1})=0, \\
                                                                          p^{e-2}, & \textrm{if\ } \tr(\gamma^{p^{\alpha}+1})\neq 0.
                                                                        \end{array}
                                                                        \right.
$$
\end{Lem}

\begin{Lem}\label{Lem12}
For $b\in \mathbb{F}_{q}^{*}$ and  $c\in \mathbb{F}_{p}^{*}$,   we have
$$
|N_{b}(0,c)|=\left\{\begin{array}{ll}
                                                                          p^{e-2}, & \textrm{if\ }  \tr(\gamma^{p^{\alpha}+1})=0, \\
                                                                          p^{e-2}-p^{m-1}, & \textrm{if\ }  \tr(\gamma^{p^{\alpha}+1})\neq 0.
                                                                        \end{array}
                                                                        \right.
$$
\end{Lem}

Now it comes to prove Theorem \ref{thm1}.
\begin{proof}
By \eqref{eq-3.2} and Lemma \ref{lem7}, we get $n_{0}=p^{e-1}-(p-1)p^{m-1}.$ Using Lemma \ref{lem11}, we have $\wt(c_{b}) \in \{(p-1)p^{e-2}, (p-1)p^{e-2}-(p-1)p^{m-1}\}. $
Because $\wt(c_{b})\neq 0 $ for each $b \in \mathbb{F}_{q}^{*},$ the dimension of $\C_{D_{0}}$ is $e.$
Suppose
\begin{align*}
b_{1}&=(p-1)p^{e-2},\\
 b_{2}&=(p-1)p^{e-2}-(p-1)p^{m-1}.
\end{align*}
Also by Lemma \ref{lem11} and the value of $n_{0}, $
 we know that
\begin{align*}
A_{b_{1}}&=p^{e-1}-(p-1)p^{m-1}-1,\\
 A_{b_{2}}&=(p-1)(p^{e-1}-p^{m-1}).
\end{align*}
Hence,  we get the Table 1. By the above two lemma  , it is easy
to get the complete weight enumerator of $\C_{D_{0}}.$ And we complete the proof of Theorem \ref{thm1}.
\end{proof}

\subsubsection{The second case}

  $a\neq0,\ m/d\equiv 1 \mod2.$

  By \eqref{eq-3.1}, \eqref{eq-3.3} and Lemmas \ref{lem7} and \ref{lem8}, it is easy to get the
values of $n_{a}, |N_{b}(a, 0)|$ and $|N_{b}(a, c)|$.
\begin{Lem}\label{lem13}
For $a\in \mathbb{F}_{p}^{*}$, if $ m/d\equiv 1 \mod2,$ then $n_{a}=p^{e-1}+p^{m-1}.$
\end{Lem}

\begin{Lem}\label{lem14}
For $b\in \mathbb{F}_{q}^{*}$ and  $a\in \mathbb{F}_{p}^{*}$,  if $m/d\equiv 1 \mod2,$ then
$$
|N_{b}(a,0)|=\left\{\begin{array}{ll}
                                                                          p^{e-2}+p^{m-1}, & \textrm{if\ }  \tr(\gamma^{p^{\alpha}+1})=0, \\
                                                                          p^{e-2}-p^{m-1}\overline{\eta}(-a\tr(\gamma^{p^{\alpha}+1})), & \textrm{if\ } \tr(\gamma^{p^{\alpha}+1})\neq 0.
                                                                        \end{array}
                                                                        \right.
$$
\end{Lem}

\begin{Lem}\label{lem15}
For $ b\in \mathbb{F}_{q}^{*}$, $a$ and $c\in \mathbb{F}_{p}^{*}$, if  $m/d\equiv 1 \mod2$, then
$$
|N_{b}(a,c)|=\left\{\begin{array}{ll}
                                                                          p^{e-2}, & \textrm{if\ }  \tr(\gamma^{p^{\alpha}+1})=0, \\
                                                                          p^{e-2}, & \textrm{if\ } \tr(\gamma^{p^{\alpha}+1})=c^{2}/(4a), \\
                                                                          p^{e-2}-p^{m-1}\overline{\eta}\left(c^{2}-4a\tr(\gamma^{p^{\alpha}+1})\right), & \textrm{otherwise}.
                                                                        \end{array}
                                                                        \right.
$$
\end{Lem}

Now we begin to prove Theorem \ref{thm2}.
\begin{proof}
In this case, using the above three lemma, as the proof of Theorem \ref{thm1} one can prove Theorem \ref{thm2} similarly. The details are omitted.
\end{proof}

\subsubsection{The third case}
  $a=0$ and $m/d\equiv 0 \mod2.$

  By \eqref{eq-3.1}, \eqref{eq-3.3}and Lemmas \ref{lem7} and \ref{lem9}, we get the following two lemmas.

\begin{Lem}\label{lem16}
Let $b\in \mathbb{F}_{q}^{*}$ and $m/d\equiv 0\mod2.$ If $\ X^{p^{2\alpha}}+ X=-b^{p^{\alpha}}$
has no solution in $\mathbb{F}_{q}$,
then
$$|N_{b}(0,0)|=p^{e-2}-(p-1)p^{m+d-2}.$$
If $X^{p^{2\alpha}}+ X=-b^{p^{\alpha}}$ has a solution $\gamma$ in $\mathbb{F}_{q},$ then
$$
|N_{b}(0,0)|=\left\{\begin{array}{ll}
                                                                          p^{e-2}-(p-1)p^{m+d-1}, & \textrm{if\ }  \tr(\gamma^{p^{\alpha}+1})=0, \\
                                                                          p^{e-2}, & \textrm{if\ } \tr(\gamma^{p^{\alpha}+1})\neq 0.
                                                                        \end{array}
                                                                        \right.
$$
\end{Lem}

\begin{Lem}\label{lem17}
Let $b\in \mathbb{F}_{q}^{*}$, $c\in \mathbb{F}_{p}^{*}$ and  $m/d \equiv 0\mod2. $ If $X^{p^{2\alpha}}+ X=-b^{p^{\alpha}} $
has no solution in $\mathbb{F}_{q}, $
then
$$|N_{b}(0,c)|=p^{e-2}-(p-1)p^{m+d-2}.$$
  If $X^{p^{2\alpha}}+ X=-b^{p^{\alpha}}$ has a solution  $ \gamma$ in $\mathbb{F}_{q},$ then
$$
|N_{b}(0,c)|=\left\{\begin{array}{ll}
                                                                          p^{e-2}, & \textrm{if\ }  \tr(\gamma^{p^{\alpha}+1})=0, \\
                                                                          p^{e-2}-p^{m+d-1}, & \textrm{if\ } \tr(\gamma^{p^{\alpha}+1})\neq 0.
                                                                        \end{array}
                                                                        \right.
$$
\end{Lem}

Now we give the proof of Theorem \ref{thm3}.

\begin{proof}
By \eqref{eq-3.1} and Lemma \ref{lem7}, we get $n_{0}=p^{e-1}-(p-1)p^{m+d-1}.$ Together with Lemma \ref{lem16}, we have that $\wt(c_{b})$ has three nonzero values.
 Suppose
\begin{align*}
 b_{1}&=(p-1)p^{e-2}-(p-1)^{2}p^{m+d-2},\\
  b_{2}&=(p-1)p^{e-2},\\
   b_{3}&=(p-1)p^{e-2}-(p-1)p^{m+d-1}.
 \end{align*}
  By Lemma \ref{lem10},
$A_{b_{1}}=p^{e}-p^{e-2d}.$  By the first two Pless Power Moments(\cite{HP03}, p. 260) the frequency
$A_{b_{i}}$ of $w_{i}$ satisfies the following equations:
\begin{eqnarray}\label{eq-3.4}
\left\{\begin{array}{l}
         A_{b_1}+A_{b_2}+A_{b_3}=p^{m}-1, \\
         b_1A_{b_1}+b_2A_{b_2}+b_3A_{b_3}=p^{e-1}(p-1)n,
       \end{array}
       \right.
\end{eqnarray}
where $n=p^{e-1}-(p-1)p^{m+d-1}-1$. Solving the equations gives the weight distribution of Table 5. By Lemma \ref{lem17}, together
the definition of complete weight enumerator of codes, we can obtain the complete weight enumerator if $ \C_{D_{0}}.$ The
proof is completed.
\end{proof}

\subsubsection{The fourth case}
 $a\neq0$ and $m/d\equiv 0 \mod2.$

 By \eqref{eq-3.1}, \eqref{eq-3.3}, Lemmas \ref{lem7} and \ref{lem9}, we have the following three lemmas.
\begin{Lem}\label{lem18}
Let $a\in \mathbb{F}_{p}^{*}$ and $m/d\equiv 0 \mod2$. Then  $n_{a}=p^{e-1}+p^{m+d-1}.$
\end{Lem}

\begin{Lem}\label{lem19}
   Let $b\in \mathbb{F}_{q}^{*}$, $a\in \mathbb{F}_{p}^{*}$ and $m/d \equiv 0\mod2.$   If $\ X^{p^{2\alpha}}+ X=-b^{p^{\alpha}}$
has no solution in $ \mathbb{F}_{q},$
then
$$|N_{b}(a,0)|=p^{e-2}+p^{m+d-2}.$$
If $X^{p^{2\alpha}}+ X=-b^{p^{\alpha}}$ has a solution $\gamma$ in $\mathbb{F}_{q},$ then
$$
|N_{b}(a,0)|=\left\{\begin{array}{ll}
                                                                          p^{e-2}+p^{m+d-1}, & \textrm{if\ } \ \tr(\gamma^{p^{\alpha}+1})=0, \\
                                                                          p^{e-2}-p^{m+d-1}\eta\left(-a\tr(\gamma^{p^{\alpha}+1})\right), & \textrm{if\ } \ \tr(\gamma^{p^{\alpha}+1})\neq 0.
                                                                        \end{array}
                                                                        \right.
$$
\end{Lem}

\begin{Lem}\label{lem20}
Let $a,\ c \in \mathbb{F}_{p}^{*}$  and $b\in \mathbb{F}_{q}^{*}.$ Let $m/d \equiv 0\mod2$.  If  $ X^{p^{2\alpha}}+ X=-b^{p^{\alpha}}$
has no solution in $ \mathbb{F}_{q},$
then
$$|N_{b}(a,c)|=p^{e-2}+p^{m+d-2}.$$
  If $X^{p^{2\alpha}}+ X=-b^{p^{\alpha}}$ has a solution $\gamma$ in $\mathbb{F}_{q},$ then
$$
|N_{b}(a,c)|=\left\{\begin{array}{ll}
                                                                          p^{e-2}, & \textrm{if\ }  \tr(\gamma^{p^{\alpha}+1})=0, \\
                                                                          p^{e-2}, & \textrm{if\ }  \tr(\gamma^{p^{\alpha}+1})=c^{2}/(4a), \\
                                                                          p^{e-2}-p^{m+d-1}\overline{\eta}\left(c^{2}-4a\tr(\gamma^{p^{\alpha}+1})\right), & \textrm{otherwise}.
                                                                        \end{array}
                                                                        \right.
$$
\end{Lem}

Now we begin to prove Theorem \ref{Thm4}.
\begin{proof}
In this case, by the three lemmas above, we can get the results in  Theorem \ref{thm3}. And we omit the details.
\end{proof}

By the definition of $\overline{\C}_{D_{a}}$, Corollaries \ref{thm1-1}, \ref{thm2-2}, \ref{thm3-3} and \ref{thm4-4} follow directly from
 Theorems  \ref{thm1}, \ref{thm2}, \ref{thm3} and \ref{Thm4}, respectively. We omit the proofs.

\section{Concluding Remarks}
In this paper, the complete weight enumerators of two families of linear code with few weights are determined. Let $w_{\min}$ and $w_{\max}$ denote the minimum and maximum nonzero weight of a linear code, respectively.   As introduced in \cite{YD06}, any linear code over $\mathbb{F}_{p}$ can be employed to construct secret sharing schemes with interesting access structures if $
\frac{w_{\min}}{w_{\max}}>\frac{p-1}{p}$. It can be verified that  the linear code in  Theorem 2, Corollary 3 and Theorem 4 have  the property $\frac{w_{\min}}{w_{\max}}>\frac{p-1}{p}$  if $m\geq3$.

\end{document}